\documentclass[11pt]{article}
\usepackage{graphicx} 
\usepackage{amsmath}
\usepackage{amsthm}
\usepackage{verbatim}
\usepackage{amssymb}
\usepackage{braket}
\usepackage{algorithm}
\usepackage{algorithmic}
\usepackage{float}
\usepackage{authblk}
\usepackage{enumitem}
\usepackage[margin=1in]{geometry}
\usepackage{chngcntr}
\usepackage[none]{hyphenat}
\newcommand{\cc}[1]{\mathsf{#1}}
\newcommand{\CliffordMIP}{\(\cc{Clifford}\)--\(\cc{MIP}^\ast\)\ }

\title{Clifford Strategies in Interactive Protocols are Classically Simulatable}
\author{Itay Shalit\thanks{Weizmann Institute of Science, Rehovot, Israel. Email: itay.shalit@weizmann.ac.il. Supported by the Horizon Europe Research and Innovation Program via ERC Project ACQUA (Grant 101087742).}}

\newtheorem{theorem}{Theorem}
\newtheorem{lemma}{Lemma}
\newtheorem{corollary}{Corollary}
\newtheorem{definition}{Definition}

\counterwithout{algorithm}{subsection}

\newcommand{\K}{\mathsf{K}}
\newcommand{\R}{\mathsf{R}}

\date{\vspace{-8ex}}
\begin{document}
\maketitle
\begin{abstract}
\noindent\(\cc{MIP}^\ast\) is the class of languages decidable by an efficient classical verifier interacting with multiple quantum provers that share entangled qubits but cannot communicate. Notably, \(\cc{MIP}^\ast\) was proved to equal RE, the class of all recursively enumerable languages. \newline
We introduce the complexity class \(\cc{Clifford}\)--\(\cc{MIP}^\ast\), which restricts quantum provers to Clifford operations and classical post-processing of measurement results, while still allowing shared entangled qubits in any quantum state. We show that any strategy in this model can be simulated by classical provers with shared random bits, and therefore admits a local hidden-variable description. Consequently, \CliffordMIP = \(\cc{MIP}\), a vastly smaller complexity class compared to \(\cc{RE}\). \newline
Moreover, we resolve an open question posed by Kalai et al.\ (STOC 2023), by showing that quantum advantage in any single-round non-local game requires at least two provers operating outside the \CliffordMIP computational model. This rules out a proposed approach for significantly improving the efficiency of quantum advantage tests that are based on compiling non-local games into single-prover interactive protocols.
\end{abstract}

\section{Introduction}
The $n$-qubit Clifford group consists of quantum operators on a $2^n$-dimensional Hilbert space and is generated by the Clifford operators: Hadamard (H), phase (S), and CNOT. When supplemented with the $T$ operator, Clifford gates become sufficient for universal quantum computation (UQC)\cite{bravyi2005universal}. The $n$-qubit Clifford group is the normalizer of the $n$-qubit Pauli group within the group of all $2^n \times 2^n$ unitary matrices. In other words, the $n$-qubit Clifford group is the set of all $2^n\times 2^n$ unitary operators that map any Pauli operator to another under conjugation.  \newline 

A quantum state $\ket{\psi}$ is said to be \textit{stabilized} by an operator $S$ if $S\ket{\psi} = \ket{\psi}$, and in that case $S$ is a \textit{stabilizer} of $\ket{\psi}$. The set of all stabilizers of a state forms a group under multiplication, known as a \textit{stabilizer group}. An $n$-qubit \textit{stabilizer state} has a stabilizer group that is generated by $n$ independent, commuting Pauli operators. Any Clifford operator $U$ maps a stabilizer state to another stabilizer state, with the new stabilizer group obtained by conjugating each generator of the original group by $U$. Furthermore, the set of all $n$-qubit stabilizer states is precisely the set of states that can be obtained by applying Clifford operators to the initial state $\ket{0}^{\otimes n}$. \newline

The Gottesman-Knill theorem establishes that a classical computer can efficiently simulate any quantum circuit that is initialized in a stabilizer state and consists solely of Clifford gates and measurements in the computational basis \cite{Gottesman1998Heisenberg}. This efficient classical simulation is performed by tracking the evolution of the stabilizer group of the quantum state, rather than the full $2^n$-dimensional state vector. The stabilizer group can be represented by a list of $n$ generators, which implies that the size of this representation is polynomial in $n$. Accordingly, tracking it can be done efficiently with a classical computer. \newline

While the Gottesman-Knill theorem provides a valuable insight into the limitations of quantum computation restricted to Clifford operations, it does not characterize the power of Clifford-based quantum computation in the non-local setting, that of interactive protocols. In that setting, a classical verifier interacts with non-communicating provers through one or more rounds of classical communication, ultimately deciding whether to accept or reject based on the interaction. Notably, quantum provers can utilize entangled qubits shared before the interaction to execute strategies that classical provers with shared randomness cannot simulate. In other words, the resulting correlations do not admit a local hidden-variable description. A well-known example is the CHSH game, a simple non-local game in which entanglement-sharing quantum provers can outperform the best classical strategies \cite{CHSH}.

\subsection{Prior Work}

Multiple studies have examined the relationship between the emergence of correlations that do not admit a local hidden-variable description and the use of non-Clifford computational resources. Many consider non-local games in which the provers are restricted to using only Clifford operations during the interaction. Assuming that the state shared by the provers is a stabilizer state or is close to being one, separations are proven between the maximal performance achievable by such restricted provers, and that of general quantum provers. We note that in this computational model, each prover's choice of Clifford operations may depend on the question it receives from the verifier. This requires using classically-controlled Clifford operations---whose action depends on the value of a classical bit---which cannot be implemented using only Clifford gates. \newline 

In~\cite{howardMaxNonlocality}, Howard analyzes the CHSH game with provers restricted to classically-controlled Clifford operations and proves that the optimal strategy in this model requires a shared non-stabilizer state. Such states are often referred to in the literature as \textit{magic states}. In contrast, Clifford-restricted provers that share only a stabilizer state have no advantage over classical strategies. Cusumano et al.~\cite{cusumano2025nonstabilizernessviolationschshinequalities} further investigate this setting. They show that the stabilizer entropy of the shared state---a quantitative measure of the non-Clifford resources involved in its preparation---is related to the ability of Clifford-restricted provers to outperform classical strategies in the CHSH game. \newline

Howard and Vala explore this scenario using the framework of magic state distillation (MSD) \cite{howardNonlocalityUniversal}. In MSD, multiple noisy copies of a non-stabilizer state are transformed into fewer, higher-fidelity copies of that state using only Clifford operations. Pure non-stabilizer states together with classically-controlled Clifford operators can be used to implement non-Clifford gates. Thus, if a certain operator can be used to generate noisy non-stabilizer states suitable for MSD, then when combined with Clifford operators, it may suffice for universal quantum computation. The authors show that any (noisy) single-qubit non-Clifford operator which, when applied to one half of a maximally entangled state and followed by Pauli measurements, enables violation of the CHSH inequality, can also be used to prepare magic states suitable for MSD. Hence, any such operation suffices for UQC when combined with Clifford operators.\newline

Finally, Zhang, Pan and Liu construct a non-local game family, in which for an infinite number of games there exists a perfect quantum strategy. Nonetheless, when provers are restricted to performing classically-controlled Clifford operations on a shared stabilizer state, their win rate is strictly less than 1~\cite{Zhang_2024}. \newline 

We contribute to this line of work by examining interactive protocols in which the provers may share an arbitrary quantum state but are restricted to performing only Clifford operations in addition to classical post-processing of measurement results during the interaction. Unlike prior work, in our setting the provers may share any quantum state, but are prohibited from performing classically-controlled Clifford operations, as implementing these operations requires using non-Clifford gates. For a discussion of the motivation behind this computational model, see Section~\ref{sec:comp_model}. \newline 

\subsection{Results}

We compare two different computational models. In the first model, which we refer to as the \emph{Classical Computational Model}, the provers can only perform classical computations and are allowed to share random bits prior to the interaction. In the second model, termed the \emph{Clifford Computational Model}, the provers are quantum but are restricted to applying only Clifford operators to their qubits, in addition to performing classical post-processing of measurement results. Moreover, in this model the provers are allowed to share entangled qubits in any state before the interaction begins. \newline 

We prove the following theorem. \newline 

\textbf{Theorem:} \textit{For any $\K, \R \in \mathbb{N}$, let G be an interactive protocol with $\K$ provers 
and $\R$ rounds, and let S be a quantum strategy for G in the Clifford Computational Model. 
Then, there exists a classical strategy $\tilde{S}$ for G which is equivalent to S.} \newline 

Strategies are considered equivalent if their interaction with the protocol's verifier produces the same distribution of questions and answers. The computation of random bits shared between the provers that execute the classical strategy we describe is not known to be efficient, unless the entangled qubits shared by the provers which execute the quantum strategy are in a stabilizer state. However, the computation performed by the classical provers throughout the interaction is efficient in terms of the circuit size of provers that execute the quantum strategy and the number of qubits they use. \newline

This result directly implies an equality between the complexity classes \CliffordMIP and $\cc{MIP}$. $\cc{MIP}$ is the class of languages decidable by a classical verifier which interacts with multiple classical provers that may sample and share random bits prior to the interaction. We define \CliffordMIP as the class of languages decidable by a classical verifier which interacts with multiple quantum provers, that share entangled qubits but are limited to performing only Clifford operations in addition to classical post-processing of measurement results. Notably, the class $\cc{MIP}^\ast$ which is defined similarly to \CliffordMIP but allows the provers to perform any quantum operations, equals the class $\cc{RE}$ of recursively enumerable languages—vastly larger than $\cc{MIP}$ \cite{mip*=re}. The stark difference between $\cc{MIP}^\ast$ and \CliffordMIP highlights the significant limitations of strategies restricted to Clifford operations in the non-local setting, compared to general quantum strategies.\newline 

Furthermore, we use our results to address a question raised by Kalai et al., in their work that introduces a method for compiling any non-local game into a single-prover interactive protocol~\cite{kalai2023quantum}. A non-local game is a \textbf{single-round} multiple-prover interactive protocol. In the compiled game, quantum homomorphic encryption (QHE) is employed to simulate the spatial separation of the different provers, and is required for evaluating the strategies of all but one of the provers on encrypted queries. In a compiled game originally played with $\K$ provers, the verifier encrypts the queries of the first $\K-1$ provers, and in each round sends one query to the prover. The prover computes its encrypted answer on the encrypted query homomorphically (without decrypting it), and the verifier decrypts the answer upon receiving it. In the $\K^{th}$ (final) round the verifier sends the query in the clear and accordingly the prover answers in the clear. Finally, the verifier decides whether to accept or reject. \newline  

It is shown that, assuming the existence of quantum homomorphic encryption satisfying a natural form of correctness with respect to auxiliary quantum input, the quantum completeness guarantee of any non-local game is preserved under this compilation method. Moreover, the classical soundness guarantee is also preserved against computationally efficient strategies. In follow-up work, similar (but weaker) results were obtained regarding the quantum soundness of compiled non-local games ~\cite{cui2024computational,kulpe2024bound,natarajan2023bounding}. When this compilation method is paired with a non-local game exhibiting quantum advantage, it naturally suggests a single-prover protocol for verification of quantum advantage. In this protocol, a verifier plays the compiled game with a single prover, who is asked to demonstrate quantum advantage by surpassing the classical soundness guarantee of the game. \newline 

This suggestion meets two of the three core requirements from protocols for verifying quantum advantage:

\begin{enumerate}
\item \emph{Provable quantum advantage}—for a non-local game in which the best possible quantum success probability is substantially higher than the classical one, quantum advantage is maintained in the compiled protocol assuming that potential adversaries may only perform efficient computations. 
\item \emph{Efficient classical verifiability}—a classical verifier can efficiently verify the success of an adversary in the protocol. For example, this can be done by running the protocol multiple times to evaluate the success rate of the adversary.
\end{enumerate}

However, the protocol might fail the third requirement---it is not necessarily realizable with near-term quantum devices.
This is due to the high computational cost of applying quantum fully-homomorphic encryption (QFHE), which allows performing any efficient quantum computation on encrypted data. This cost renders the resulting protocol rather inefficient. \newline 

In contrast, there exists a QHE scheme handling only Clifford gates which is based on quantum one-time pad and is much simpler than QFHE protocols; in particular, it only requires applying the intended Clifford gates and additional classical computations. Therefore, the authors ask whether there exists a non-local game where, once an appropriate bipartite state is prepared, all but one of the prover strategies can be implemented using only Clifford gates. If such a game exists, it could be compiled using a QHE scheme handling only Clifford gates, resulting in a truly efficient protocol that meets all three conditions.\newline 

We provide a negative answer to this question by proving the following statement. \newline

\textbf{Theorem:} \textit{Let there be a non-local game G and a quantum strategy $S$, such that the success probability of $S$ in $G$ exceeds that of any classical strategy. Then, at least two provers in $S$ must operate outside the Clifford Computational Model}. \newline

Operating outside the Clifford Computational Model means using quantum operations other than Clifford operations and classical post-processing of measurement results. We achieve this result by considering any quantum strategy for a non-local game in which all players except at most one operate in the Clifford Computational Model, and describing an equivalent classical strategy for the game. \newline 

\subsection{Discussion} \label{sec:comp_model}

In prior work that considers the power of Clifford strategies in non-local games, these strategies are generally defined as follows:
\begin{enumerate}
    \item The state shared by the provers is a stabilizer state or is close to being one. A common quantitative measure of the degree to which a state is non-stabilizer, is stabilizer entropy (SE)~\cite{cusumano2025nonstabilizernessviolationschshinequalities}.
    \item Throughout the interaction, the provers may perform classically-controlled Clifford operators, even though implementing such controlled operators requires using non-Clifford gates. This enables each prover to adapt its Clifford operators to the question it receives from the verifier.
\end{enumerate}

The \CliffordMIP computational model defined above differs from computational models considered in prior work in two respects:
\begin{enumerate}
\item The provers are allowed to share any quantum state prior to the interaction. 
\item Throughout the interaction, the provers are limited to performing Clifford operations and classical post-processing of measurement results, which \textbf{does not} enable them to utilize classically-controlled Clifford gates. 
\end{enumerate}

One motivation for our choice of computational model is its relevance to cryptographic applications. Kalai et al.\ pose the question of whether there exists a non-local game in which quantum advantage can be achieved although all provers but one are restricted to Clifford operations. Such a game could then be compiled using a highly-efficient QHE protocol, even if classical post-processing of measurement results is allowed. This would result in a truly efficient protocol for testing quantum advantage. \newline 

However, we are not aware of any simple quantum homomorphic encryption (QHE) scheme that supports classically-controlled Clifford operations. Notably, the construction of quantum fully-homomorphic encryption (QFHE) can be reduced to building a QHE scheme that supports such operations~\cite{brakerski2018qfhe, mahadev2020qfhe}. It follows that an efficient QHE scheme for classically-controlled Clifford operations would imply the existence of a similarly efficient QFHE scheme. This would arguably be considered a major breakthrough. \newline 

Therefore, in considering the possibility of a quantum strategy in a non-local game which achieves quantum advantage and yet can be efficiently compiled using the method of Kalai et al., we must prohibit the provers from using classically-controlled Clifford operations. Nevertheless, since the quantum one-time pad can be used to compile any Clifford strategy---regardless of the state of the qubits shared by the provers---we do not need to place restrictions on that state. \newline

Another motivation for choosing the \CliffordMIP computational model is that in the setting of Clifford-limited provers, it is the strongest non-universal model that allows the provers to share qubits in an arbitrary quantum state. This is because magic states, which may be shared by the provers, can be used together with controlled Clifford operations to implement $T$ gates~\cite{howardNonlocalityUniversal}. In other words, if provers may share qubits in any state, it becomes crucial to restrict them from using classically-controlled Clifford gates to avoid collapsing the model into one capable of universal quantum strategies. \newline 

Additionally, the \CliffordMIP computational model can be viewed as a natural analogue of the Clifford quantum computation model considered in the Gottesman-Knill theorem, adapted to the setting of non-local games. However, while the Gottesman-Knill theorem applies only to circuits initialized in stabilizer states, in \CliffordMIP the provers may share qubits in arbitrary quantum states. Overall, our results can be interpreted as a generalization of the Gottesman-Knill theorem to the non-local setting. \newline 

We note that unlike prior work that focuses on specific non-local game families, our results apply to all non-local games and multiple-round interactive protocols. Overall, our work advances the understanding of the connection between the emergence of correlations that do not admit a local hidden-variable description and the use of non-Clifford computational resources. Our proof provides
insight into the mechanisms underlying the classical simulatability of prover strategies in the \CliffordMIP computational model (see \ref{sec:proof_outline} for details). Furthermore, by providing an answer to the question raised by Kalai et al., we contribute to a better understanding of potential paths towards designing realizable protocols for verification of quantum advantage. 

\subsection{Future Work}

\begin{enumerate}

\item \(\cc{MIP}^\ast=\cc{RE}\), but we prove that \CliffordMIP\(=\cc{MIP}\). That is, when provers are not allowed to use $T$ gates and are limited to using Clifford gates and classical post-processing of measurement results, the decision power of quantum interactive protocols is equivalent to that of classical interactive protocols. It is left to understand how the decision power of quantum interactive proofs scales with the number of $T$ gates the provers are allowed to utilize, when this number is higher than zero but still is bounded.

\item Our work establishes that for any quantum strategy within the Clifford Computational Model for a multiple-round interactive protocol, as well as any quantum strategy for a single-round interactive protocol where at most one prover is outside the Clifford Computational Model, there exist equivalent classical strategies. This result is optimal for single-round protocols, as there are known protocols where a strategy involving two provers outside the Clifford Computational Model outperforms the best possible classical strategy. A famous example of such a protocol is the CHSH game. However, it is still needed to determine whether a multiple-round protocol exists such that a quantum strategy, with exactly one prover outside the Clifford Computational Model, surpasses the best classical strategy.

\item We raise the question of whether there exists a limited computational model distinct from Clifford-based models, for which a highly-efficient QHE scheme can be devised, and which still suffices to achieve quantum advantage in a non-local game.

\end{enumerate}

\subsection{Proof Outline} \label{sec:proof_outline}

We provide an overview of the proof of our main result. The proofs of our other results build on similar techniques. Our main result can be phrased informally as follows. \newline 

\textbf{Theorem:} \textit{Let $G$ be an interactive protocol with $\K$ provers and $\R$ rounds, and let $S$ be a quantum strategy for $G$ in the Clifford Computational Model. Then, there exists a classical strategy $\tilde{S}$ for $G$ which is equivalent to $S$.
} \newline 

For a formal statement of this result, see Theorem \ref{thm: multi_round_clifford_classical}. In our proof, given an interactive protocol $G$ and a quantum strategy $S$, we construct an equivalent classical strategy $\tilde{S}$ for $S$, provided that there exists an implementation of $S$ where all provers use only Clifford operations and classical post-processing of measurement results. The classical strategy $\tilde{S}$ operates as follows:

\begin{enumerate}
    \item A series of arbitrary queries which corresponds to a full interaction of the verifier in G with the provers, is hard-coded into the classical strategy.
    \item Before the interaction begins, the classical provers jointly sample answers from the distribution that arises if provers which execute the quantum strategy $S$ receive the hard-coded questions. The provers executing $\tilde{S}$ share an encoding of this precomputed interaction.
    \item During the interaction, the verifier sends questions that may differ from the hard-coded ones, and each classical prover computes corrections to their precomputed measurement results based on the differences between the hard-coded and received questions. Returning these corrected answers results in a strategy that is equivalent to S.
\end{enumerate}

Consider a circuit implementation $C$ of the strategy $S$. For simplicity, let us assume that for any prover $i\in[\K]$ and for any round $r\in [\R]$, the computation of the $i^{th}$ prover which executes $C$ in the $r^{th}$ round consists of applying a Clifford unitary operation $U_i^r$ followed by a measurement of a subset of the qubits held by the $i^{th}$ prover, whose result serves as the prover's answer to the verifier. Our detailed proof accounts for the general case of strategies that may include intermediate measurements.\newline

We shall describe an interaction between  provers executing the classical strategy $\tilde{S}$ and the verifier. Let us define the following notation, where subscripts index the different provers and superscripts index the rounds of interaction:
\begin{enumerate}
\item $q:=\left(\left(q_i^r\right)_{i=1}^\K\right)_{r=1}^\R$ — the sequence of questions that are hard-coded into $\tilde{S}$.
\item $a:=\left(\left( a_i^r\right)_{i=1}^{\K}\right)_{r=1}^\R$ — the answers sampled by the provers before the interaction. 
\item $\tilde{q}:=\left(\left(\tilde{q}_i^r \right)_{i=1}^\K\right)_{r=1}^\R$ — the questions sent by the verifier throughout the interaction.
\item $\tilde{a}:=\left(\left(\tilde{a}_i^r\right)_{i=1}^\K\right)_{r=1}^\R$ — the answers returned by the provers throughout the interaction.
\end{enumerate}

For any $i\in [\K]$, we define the following sequence of operators by recursion on $r\in[\R]$, as follows.

\[
R_{i}^r = U_{i}^r \left( X_{\tilde{q}_i^r \oplus q_i^r} \otimes R_{i}^{r-1} \right) U_{i}^{r\dagger},
\]

where:
\begin{itemize}
    \item $X_{\tilde{q}_i^r \oplus q_i^r}$ is a tensor product of Pauli-$X$ and Pauli-$I$ operators, that accounts for the differences between the expected and received questions. That is, $X_{\tilde{q}_i^r \oplus q_i^r}\ket{\tilde{q}_i^r}=\ket{q_i^r}$.
    \item For any $i\in[\K]$, $R_{i}^{0}$ is the identity operator $I$.
\end{itemize}

The Clifford group is the normalizer of the Pauli group within the group of all unitary operators. That is, when a Pauli operator is conjugated by a Clifford operator, the result is a Pauli operator. For any $i\in[\K], r\in[\R]$, $U_i^r$ is a Clifford operator. Consequently, by performing an induction over $r$, one can verify that $R_i^r$ is a Pauli operator. \newline 

The $i^{th}$ prover computes its answer $\tilde{a}_i^r$ in the $r^{th}$ round as follows.
\[\ket{\tilde{a}_i^r}\bra{\tilde{a}_i^r}\otimes I
= R_i^r(\ket{a_i^r}\bra{a_i^r}\otimes I)R_i^{r\dag}\]

Here, $I$ is the identity operator which operates on the qubits that the prover does not measure to obtain its answer. We note that a Pauli operator indeed maps any standard basis element to another standard basis element, and potentially introduces a phase. Therefore, the above relation is well-defined. Crucially, each prover can compute its corresponding operators locally—without exchanging information with the other provers. It follows that the classical strategy respects the locality constraints of the protocol. \newline

The following equality plays a major role in the mechanism that enables these computations to be performed locally: 
\begin{align}
\label{eq:commutation}
U_{i}^r \left( X_{\tilde{q}_i^r \oplus q_i^r} \otimes R_{i}^{r-1} \right) = R_i^r U_i^r
\end{align}

Loosely speaking, this equality means that first applying an operator that maps between the query $q_i^r$ and the query $\tilde{q}_i^r$ and then applying $U_i^r$, is equivalent to applying $U_i^r$ and then another Pauli operator, $R_i^r$. Rather than mapping one query to another, $R_i^r$ is used to map one answer to another, as can be seen by examining the simple case of a single-round, two-prover game: 
\begin{align}
\label{eq:replace}
\Pr_S[a_1,a_2 \mid q_1, q_2] &= |\bra{a_1,a_2} (U_1\otimes U_2)\ket{q_1}\ket{q_2}\ket{\psi}|^2 \notag \\
&= | \bra{a_1,a_2} U_1 (X_{q_1\oplus \tilde{q}_1}\otimes I)\otimes U_2 (X_{q_2\oplus \tilde{q}_2}\otimes I)\ket{\tilde{q}_1}\ket{\tilde{q}_2}\ket{\psi}|^2 \notag \\
&= |\bra{a_1,a_2}(R_1 U_1\otimes R_2 U_2)\ket{\tilde{q}_1}\ket{\tilde{q}_2}\ket{\psi}|^2 \notag \\
&= |\bra{\tilde{a}_1,\tilde{a}_2} (U_1\otimes U_2)\ket{\tilde{q}_1}\ket{\tilde{q}_2}\ket{\psi}|^2 \notag \\
&= \Pr_S[\tilde{a}_1,\tilde{a}_2 \mid \tilde{q}_1, \tilde{q}_2]
\end{align}

where $U_1$ operates on the qubits held by the first prover, and $U_2$ on those held by the second prover. We assume here without loss of generality that each of the provers measures all their qubits as the last step of their computation, and the result serves as their answer to the verifier. That is, for any $i \in [\K]$, $R_i\ket{a_i}=\ket{\tilde{a}_i}$. The third equality above holds due to Eq.\ \eqref{eq:commutation}. Denoting by $\pi(\tilde{q}_1, \tilde{q}_2)$ the probability that the verifier asks the question $(\tilde{q}_1,\tilde{q}_2)$, we have:
\begin{align}
&\Pr_{\tilde{S}}[\tilde{q}_1,\tilde{q}_2,\tilde{a}_1,\tilde{a}_2] = \pi(\tilde{q}_1, \tilde{q}_2)\cdot \Pr_{\tilde{S}}[\tilde{a}_1,\tilde{a}_2 \mid \tilde{q}_1,\tilde{q}_2]
\end{align}

Recall that the answer of the $i^{th}$ prover is $R_i \ket{a_i}$, where $a_i^r$ is the corresponding answer in the precomputed interaction and $R_i$ depends on the question sent to the $i^{th}$ prover by the verifier. Therefore, assuming both provers which execute $\tilde{S}$ received the questions $(\tilde{q}_1,\tilde{q}_2)$, their answers are $(\tilde{a}_1,\tilde{a}_2)$ if and only if the answers of the provers in precomputed interaction are $(a_1, a_2)$. Hence, we may exchange  $\Pr_{\tilde{S}}[\tilde{a}_1,\tilde{a}_2 \mid \tilde{q}_1,\tilde{q}_2]$ with $\Pr_{S}[a_1,a_2 \mid q_1,q_2]$ to obtain: 
\begin{align}
&\Pr_{\tilde{S}}[\tilde{q}_1,\tilde{q}_2,\tilde{a}_1,\tilde{a}_2] = \pi(\tilde{q}_1, \tilde{q}_2)\cdot \Pr_{S}[a_1,a_2 \mid q_1,q_2]
\end{align}

Eq.\ \eqref{eq:replace} states that $\Pr_S[a_1,a_2 \mid q_1, q_2]=\Pr_S[\tilde{a}_1,\tilde{a}_2 \mid \tilde{q}_1, \tilde{q}_2]$, and therefore, 
\begin{align}
&\Pr_{\tilde{S}}[\tilde{q}_1,\tilde{q}_2,\tilde{a}_1,\tilde{a}_2] = \pi(\tilde{q}_1, \tilde{q}_2)\cdot \Pr_{S}[\tilde{a}_1,\tilde{a}_2 \mid \tilde{q}_1,\tilde{q}_2]=\Pr_{S}[\tilde{q}_1,\tilde{q}_2,\tilde{a}_1,\tilde{a}_2]
\end{align}

That is, the probability of the classical provers producing any given interaction with the verifier, is equal to the probability of quantum provers producing this interaction. In other words, $\tilde{S}$ and $S$ are equivalent strategies. \newline

Next, we briefly discuss how to extend this result to the case of multiple rounds and provers. For all $r\in[\R]$, we denote by $\rho_{q,a,r}$ the post-measurement state of all qubits held by the $\K$ provers which execute the quantum strategy $S$, right after all provers return the answers of the $r^{th}$ round to the verifier, assuming the following conditions hold:

\begin{itemize}
    \item For any $j \leq r$, the provers receive the queries $(q_1^j, \dots, q_K^j)$ from the verifier in the $j^\text{th}$ round.
    \item For any $j\leq r$, the answers returned by the provers in the $j^{th}$ round are $a^j = \left( a_i^j \right)_{i=1}^\K$.   
\end{itemize}

$\rho_{\tilde{q},\tilde{a},r}$ is defined analogously. We prove by direct calculation that the following equality holds:

\begin{equation}
\label{eq:state_link_introduction}
\forall r\in[\R]: \quad \rho_{\tilde{q}, \tilde{a}, r} = \left( \bigotimes_{i=1}^\K R_{i}^r \right) \rho_{q, a, r} \left( \bigotimes_{i=1}^\K R_{i}^r \right)^\dagger
\end{equation}

This invariance property enables us to derive an equivalence that is the generalization of Eq.\ \eqref{eq:replace} for multiple rounds and provers:
\begin{align}
\label{eq:multi_round_replace}
&\forall r\in[\R]: \Pr_{S}[\tilde{a}^r \mid \tilde{q}^1,\tilde{a}^1,...,\tilde{q}^{r-1},\tilde{a}^{r-1}, \tilde{q}^r]=\Pr_{S}[a^r \mid q^1,a^1,...,q^{r-1},a^{r-1}, q^r]
\end{align}

That is,  
for each round \(r\in[\R]\), the quantum strategy \(S\) assigns the 
same probability to the classical provers’ answers \(\tilde{a}^r\) as it 
does to the precomputed answers \(a^r\). Here, \(\Pr_{S}[a^r \mid q^1,a^1,\dots,q^{r-1},a^{r-1},q^r]\) is the probability 
that quantum provers executing \(S\) produce the answer \(a^r\) in round~\(r\), 
given the interaction history \((q^1,a^1,\dots,q^{r-1},a^{r-1},q^r)\), where $a^j=(a_1^j,...,a_\K^j)$, $q^j=(q_1^j,...,q_\K^j)$. \newline

Let us elaborate on the meaning of Eq.\ \eqref{eq:multi_round_replace}:
\begin{enumerate}
\item Consider the answer \(\tilde{a}^r\) that the classical provers 
\(\tilde{S}\) return in round~\(r\) after receiving the query \(\tilde{q}^r\). 
The left-hand side of Eq.\ ~\eqref{eq:multi_round_replace} is the probability that 
quantum provers (under strategy \(S\)) also return \(\tilde{a}^r\) in 
round~\(r\), assuming they had followed the \emph{same} interaction 
history 
\(\bigl(\tilde{q}^1,\tilde{a}^1,\dots,\tilde{q}^{r-1},\tilde{a}^{r-1},\tilde{q}^r\bigr)\) as the classical provers did.
\item Consider $a^r$, the answer in the $r^{th}$ round of the \emph{precomputed} interaction that was sampled by the provers executing $\tilde{S}$ before the interaction with the verifier. The right-hand side of Eq. \eqref{eq:multi_round_replace} is the probability that quantum provers executing strategy $S$ return the answer $\tilde{a}^r$ in round $r$, assuming that their interaction history is $(q^1,a^1,...,q^{r-1},a^{r-1}, q^r)$—exactly as in the precomputed interaction.

\end{enumerate}

Since these two probabilities coincide for every round, at each round of the interaction the classical strategy \(\tilde{S}\) effectively reproduces the outcome distribution of the quantum strategy \(S\)—a probability distribution from which the precomputed interaction is sampled (albeit with hard-coded questions). Eq.\ \eqref{eq:multi_round_replace} enables us to complete the proof for multiple-round protocols analogously to the way we used Eq.\ \eqref{eq:replace} to prove the single-round case. \newline 

Regarding the efficiency of the simulation, Pauli operators can be efficiently represented classically. Additionally, the mapping of one Pauli operator to another via a Clifford operator, as well as the mapping of one standard basis element to another via a Pauli operator, can both be efficiently computed classically. This ensures that except the pre-computation of the shared encoded interaction,
the computation performed by the classical provers throughout the interaction with the verifier is efficient in terms of the circuit size of provers that apply the corresponding quantum strategy and the number of qubits they use. 
Sampling the answers before the interaction is only guaranteed to be efficient if the quantum provers share a stabilizer state.

\subsection{Acknowledgements}
Thanks are due to Zvika Brakerski and Thomas Vidick for helpful discussions and valuable insights, which contributed to the development of this work.

\section{Preliminaries}

We assume basic knowledge of the theory of quantum information and computation. For an introduction to this topic, see \cite{nielsenClifford}. We open with a series of definitions regarding interactive protocols.

\begin{definition}
An \emph{interactive protocol} with \( \K \) provers and \( \R \) rounds is defined as
\[
G = \left(  Q,  A ,\left\{ \pi_r \right\}_{r=1}^\R,\ v \right).
\]
Here:
\begin{itemize}
    \item $Q$ is a finite set of possible queries that can be sent to any prover, and $A$ is a finite set of possible answers that any prover can return to the verifier.
    \item The protocol proceeds in \( \R \) rounds. In each round \( r \) (\( r\in[\R] \)):
    \begin{itemize}
        \item The verifier sends a query \( q_i^r \in Q \) to each prover \( i \).
        \item Each prover \( i \) replies with an answer \( a_i^r \in A \).
    \end{itemize}
    \item The provers cannot communicate with each other during the protocol.

    \item A \emph{history} up to round \( r \) (\( r\in[\R] \)) is a sequence
    \[
    h^r := \left(\left( q_i^j, a_i^j\right)_{i=1}^\K\right)_{j=1}^r
    \]

    consisting of the queries sent to and answers received from all provers in rounds \( 1 \) to \( r \). Let us denote the set of histories up to round $r$ by $H^r$. \newline 

    We denote all questions sent by the verifier in round $j$ $(j\in [r])$ of $h^r$ by
    \[q^j := (q_i^j)_{i=1}^\K\]

    Similarly, we denote all answers sent by the provers in round $j$ $(j\in [r])$ of $h^r$ by
    \[a^j := (a_i^j)_{i=1}^\K\]   

    \item For each prover \( i \), a \emph{local history} up to round \( r \) is a sequence
    \[
    h_i^r = \left( q_i^j, a_i^j \right)_{j=1}^r,
    \]
    consisting of the questions received by prover \( i \) and the answers provided by prover \( i \) in rounds \( 1 \) to \( r \). Let us denote the set of local histories for prover $i$ up to round $r$ by $H_i^r$.

    \item For each round \( r \), the function \( \pi_r \) assigns a probability mass to the sequence of questions sent by the verifier in that round, conditioned on the \emph{history} up to that point:
    \[
    \pi_r : H^{r-1} \times Q^\K \rightarrow [0,1],
    \]

    \item The function \( v \) determines the outcome of the protocol (accept or reject) based on the \emph{history} after \( R \) rounds:
    \[
    v: H^\R \rightarrow \{0,1\},
    \]
    where \( H^\R \) is the set of possible histories after \( \R \) rounds.
\end{itemize}

\end{definition}

\begin{definition}
A non-local game is an interactive protocol with a single round of communication.
\end{definition}

\begin{definition}
A \emph{quantum strategy} \( S \) for an interactive protocol \( G = \left(Q, A, \left\{ \pi_r \right\}_{r=1}^\R, v \right) \) is defined as
\[
S = \left( \ket{\psi},\ \left(\left( M_{i}^r \right)_{i=1}^\K \right)_{r=1}^\R\right).
\]

Here, \( \ket{\psi} \) is a pure quantum state over \( n \) qubits (for some \( n \in \mathbb{N} \)) which are shared by the provers. \( M_{i}^r \) (\( i\in[\K], \ r\in[\R] \)) is the projective measurement (PVM) performed by the \( i^\text{th} \) prover in the \( j^\text{th} \) round, on the question (bit-string) it receives from the verifier, and on its part of \( \ket{\psi} \). The outcomes of \( M_{i}^j \) correspond to elements of \( A \).
\end{definition}

Some definitions of a strategy allow a prover to operate on ancilla qubits in addition to the prover's part of $\ket{\psi}$. However, we assume without loss of generality that all ancilla qubits a prover operates on are included in the part of $\ket{\psi}$ that is held by the prover.

\begin{definition}
A \emph{classical strategy} with shared randomness \( S \) for \( G \) is defined as

\[
S = \left( P,\ \left(\left( f_{i}^r\right)_{i=1}^\K\right)_{r=1}^\R \right).
\]
Here, \( P \) is a probability distribution on \( \{0,1\}^\lambda \) (for some $\lambda \in \mathbb{N}$), from which the random bits that are shared by the provers are sampled at the beginning of the protocol. For each \( i\in[\K], \ r\in[\R] \):
\begin{itemize}
    \item \( H_i^{r-1} \) denotes the set of possible local histories of prover \( i \) up to round \( r-1 \).
    \item \( f_{i}^r : \{0,1\}^\lambda \times H_i^{r-1} \times Q \rightarrow A  \) is a function that, given a shared random bit-string \( s \in \{0,1\}^\lambda \), a local history \( h_{i}^{r-1} \in H_{i}^{r-1} \), and a question \( q_i^r \in Q \) received in round \( r \), outputs the answer \( a_i^r \in A \) of prover \( i \) in round \( r \).
\end{itemize}
\end{definition}

\begin{definition}
Let \( G \) be an interactive protocol, and let $S = \left( \ket{\psi},\ \left( M_i\right)_{i=1}^\K\right)$ be a quantum strategy for it. An \emph{implementation} of \( S \) is a sequence of quantum circuits $\left(\left( C_i^r \right)_{i=1}^\K\right)_{r=1}^\R$, such that for each prover \( i \) and each round \( r \), \( C_i^r \) implements the PVM \( M_i^r \).
\end{definition}

\begin{definition}
For any strategy \( S \) for \( G \), whether quantum or classical, we denote by \( \Pr_{S} \) the joint probability distribution over histories \( h \) after \( R \) rounds.

Formally, for any history \( h := \left( \left( q_i^r, a_i^r \right)_{i=1}^\K \right)_{r=1}^\R \):
\[
\Pr_{S} \left[ h^R \right] = \prod_{r=1}^\R \left( \pi_r \left( h^{r-1}, q^r \right) \cdot \Pr_S \left( a^r \mid h^{r-1}, q^r\right) \right),
\]
where:
\begin{itemize}
    \item \( h^{r-1} \) is the history up to round \( r-1 \).
    \item \( \Pr_S \left[a^r \mid h^{r-1}, q^r\right]\) is the probability that the provers following strategy \( S \), provide answers \( a_1^r, \dots, a_K^r \) in round \( r \) given the history \( h^{r-1} \) and the questions \( q_1^r, \dots, q_K^r \) in round \( r \).
    
\end{itemize}
\end{definition}

\begin{definition}
Two strategies \( S \) and \( S' \) for \( G \), whether quantum or classical, are considered \emph{equivalent} if and only if they induce the same joint probability distribution over interaction histories.
\end{definition}

\begin{definition}
\label{def:PauliGroup}\cite{nielsenClifford}

The single-qubit \emph{Pauli group} $\mathcal{P}_1$ consists of all $2 \times 2$ unitary matrices generated by the Pauli matrices $I$, $X$, $Y$, and $Z$, along with multiplicative factors $\{\pm 1, \pm i\}$:
\[
\mathcal{P}_1 = \left\{ e^{i \theta} P \mid \theta \in \left\{0, \tfrac{\pi}{2}, \pi, \tfrac{3\pi}{2}\right\},\ P \in \{I, X, Y, Z\} \right\},
\]
where
\[
I = \begin{pmatrix} 1 & 0 \\ 0 & 1 \end{pmatrix},\quad
X = \begin{pmatrix} 0 & 1 \\ 1 & 0 \end{pmatrix},\quad
Y = \begin{pmatrix} 0 & -i \\ i & 0 \end{pmatrix},\quad
Z = \begin{pmatrix} 1 & 0 \\ 0 & -1 \end{pmatrix}.
\]
The $n$-qubit Pauli group $\mathcal{P}_n$ is defined as the $n$-fold tensor product of single-qubit Pauli operators with phase factors:
\[
\mathcal{P}_n = \left\{ e^{i \theta} P_1 \otimes P_2 \otimes \cdots \otimes P_n\ \mid \ \theta \in \left\{0, \tfrac{\pi}{2}, \pi, \tfrac{3\pi}{2}\right\},\ P_i \in \{I, X, Y, Z\} \right\}.
\]
\end{definition}.

\begin{definition} The set of Clifford quantum gates consists of the following unitary operations:

\[
\begin{array}{ccc}
H = \frac{1}{\sqrt{2}} \begin{pmatrix} 1 & 1 \\ 1 & -1 \end{pmatrix} &
\text{CNOT} = \begin{pmatrix}
1 & 0 & 0 & 0 \\
0 & 1 & 0 & 0 \\
0 & 0 & 0 & 1 \\
0 & 0 & 1 & 0
\end{pmatrix} &
S = \begin{pmatrix}
1 & 0 \\
0 & i
\end{pmatrix}
\end{array}
\]

\end{definition}

\begin{definition}
\label{def:Clifford_map}
Let $U(2^n)$ denote the group of $2^n \times 2^n$ unitary operators acting on $n$ qubits. The \emph{Clifford group} on $n$ qubits, denoted by $\mathcal{C}_n$, is the normalizer of the Pauli group $\mathcal{P}_n$ in $U(2^n)$:
\[
\mathcal{C}_n = \left\{ U \in U(2^n) \mid U \mathcal{P}_n U^\dagger = \mathcal{P}_n \right\}
\]
That is, for all $U \in \mathcal{C}_n$ and all $P \in \mathcal{P}_n$, the conjugation of $P$ by $U$ yields another Pauli operator:
\[
U P U^\dagger \in \mathcal{P}_n
\]
Furthermore, the Clifford group is the largest subgroup of $U(2^n)$ that preserves the Pauli group under conjugation.
\end{definition}

Next, let us state a lemma regarding the Clifford group.

\begin{lemma}
\cite{nielsenClifford} \label{lemma:clifford}
The Clifford group on $n$ qubits, denoted by $\mathcal{C}_n$, is generated by the following set of matrices:

\[\{H_i: i\in [n]\} \cup \{S_i: i\in [n]\} \cup \{CNOT_{i,j}: i,j \in [n], i\neq j\}\]

Where $H_i, S_i$ are $H,S$ (respectively) operating on the $i^{th}$ qubit, and $CNOT_{i,j}$ is $CNOT$ operating on the $i^{th},j^{th}$ qubits.
\end{lemma}

According to Definition \ref{def:Clifford_map}, conjugation by any Clifford unitary \( U \in \mathcal{C}_n \) maps elements of the Pauli group \( \mathcal{P}_n \) to other elements within \( \mathcal{P}_n \). Lemma \ref{lemma:clifford} establishes that any such \( U \) can be implemented using Clifford gates. To classically compute \( U P U^\dagger \) for \( P \in \mathcal{P}_n \), one can proceed as follows:

\label{explanation:clifford_pauli_update}

\begin{enumerate}
    \item Represent \( P \) as a tensor product of single-qubit Pauli operators and a complex phase.
    \item Decompose \( U \) into a sequence of Clifford gates (e.g., \( H \), \( S \), and CNOT gates).
    \item Sequentially update the representation of \( P \) by applying the conjugation action of each gate in the sequence. Since the conjugation of a Pauli operator by a Clifford gate results in another Pauli operator, the updated representation remains within \( \mathcal{P}_n \).
\end{enumerate}

The computational complexity depends on the number of gates \( m \) in the implementation of \( U \):

\begin{itemize}
    \item If \( m = O(n) \), the computation can be performed in \( O(n) \) time.
    \item Otherwise, the computation requires \( O(m) \) time.
\end{itemize}

\begin{definition}
A quantum circuit limited to Clifford gates and classical operations is one that includes only the following gates:
\begin{enumerate}
    \item Clifford gates (CNOT, H, S). 
    \item Any classical logical operations, which can only be performed on qubits immediately after they are measured. These operations are not necessarily invertible.   
\end{enumerate}
\end{definition}

\textbf{Definition ($\cc{MIP}$):} The complexity class $\cc{MIP}$ consists of all languages $L$ for which there exists a probabilistic polynomial-time verifier $V$ that interacts with multiple computationally unbounded but non-communicating classical provers $P_1, P_2, \dots, P_\K$, that are allowed to share random bits before the interaction,  such that there exist functions \( c, s : \mathbb{N} \rightarrow [0,1] \) satisfying \( c(|x|) - s(|x|) \geq \frac{1}{\operatorname{poly}(|x|)} \), and:

\begin{itemize} \item (\textbf{Completeness}) If $x \in L$, then there exist provers that convince the verifier to accept with a probability at least $c(|x|)$. \item (\textbf{Soundness}) If $x \notin L$, then any provers can convince the verifier to accept except with a probability at most $s(|x|)$ \end{itemize}

\vspace{1em}

\textbf{Definition ($\cc{MIP}^\ast$):} The complexity class $\cc{MIP}^\ast$ consists of all languages $L$ for which there exists a probabilistic polynomial-time verifier $V$ that interacts with multiple computationally unbounded but non-communicating quantum provers $P_1, P_2, \dots, P_\K$ that are allowed to share entangled qubits prior to the interaction, such that there exist functions \( c, s : \mathbb{N} \rightarrow [0,1] \) satisfying \( c(|x|) - s(|x|) \geq \frac{1}{\operatorname{poly}(|x|)} \), and::

\begin{itemize} \item (\textbf{Completeness}) If $x \in L$, then there exist provers that convince the verifier to accept with a probability at least $c(|x|)$. \item (\textbf{Soundness}) If $x \notin L$, then any provers can convince the verifier to accept except with a probability at most $s(|x|)$ \end{itemize}

We define a new complexity class, denoted \CliffordMIP, as follows:

\vspace{1em}

\textbf{Definition (\CliffordMIP):} The complexity class \CliffordMIP consists of all languages $L$ for which there exists a probabilistic polynomial-time verifier $V$ that interacts with multiple computationally unbounded but non-communicating quantum provers $P_1, P_2, \dots, P_\K$ that are allowed to share entangled qubits prior to the interaction, and are constrained to applying only Clifford operations in addition to classical post-processing of measurement results, such that for some parameters $c(|x|) > s(|x|) + \frac{1}{\operatorname{poly}(|x|)}$:

\begin{itemize} \item (\textbf{Completeness}) If $x \in L$, then there exist provers that convince the verifier to accept with a probability at least $c(|x|)$. \item (\textbf{Soundness}) If $x \notin L$, then any provers can convince the verifier to accept except with a probability at most $s(|x|)$. \end{itemize}

\section{Clifford Strategies in Interactive Protocols}

\subsection{Classical Simulation of Clifford Strategies}

We state the main theorem of the current section.

\begin{theorem}
\label{thm: multi_round_clifford_classical}
Let there be an interactive protocol with \( \K \) provers and \( \R \) rounds \newline $G = \left(Q , A ,\ \left\{ \pi_r \right\}_{r=1}^\R,\ v \right)$, and a quantum strategy $S =  \left( \ket{\psi},\ \left(\left( M_{i}^r \right)_{i=1}^\K \right)_{r=1}^\R\right)$ for $G$. Furthermore, there exists an implementation of $S$ denoted $C = \left(\left( C_i^r \right)_{i=1}^\K\right)_{r=1}^\R$ such that $\forall i\in[\K], r\in[\R]$, $C_i^r$ is a quantum circuit limited to Clifford gates and classical post-processing of measurement results. Then, there exists a classical strategy $\tilde{S}$ for $G$ that meets the following conditions:

\begin{enumerate}
    \item $\tilde{S}$ is equivalent to $S$.
    \item The computation performed by the provers executing $\tilde{S}$ throughout the interaction is efficient in terms of the circuit size of the provers that implement $S$ and the number of qubits they use. However, computing the random bits that are shared between the provers before the interaction is only known to be efficient in terms of the same parameters if $\ket{\psi}$ is a stabilizer state. 
\end{enumerate}

\end{theorem}

To simplify the proof of Theorem \ref{thm: multi_round_clifford_classical}, we begin by stating and proving a weaker theorem. While Theorem \ref{thm: multi_round_clifford_classical} states that any Clifford strategy for an interactive protocol can be classically simulated, Lemma \ref{lemma: multi_round_clifford_unitary_classical} states that for a subset of Clifford strategies. This is the subset of strategies such that the computation performed by any prover at each round of interaction can be implemented with a Clifford unitary followed by a standard basis measurement of some of the qubits on which the unitary operates.  

\begin{lemma}
\label{lemma: multi_round_clifford_unitary_classical}
Let there be an interactive protocol with \( K \) provers and \( R \) rounds \newline $G = \left( Q, A, \left\{ \pi_r \right\}_{r=1}^\R, v \right)$ with $K$ provers, and a quantum strategy $S =  \left( \ket{\psi},\ \left(\left( M_{i}^r \right)_{i=1}^\K \right)_{r=1}^\R\right)$ for $G$. Furthermore, there exists an implementation of $S$ denoted by $C = \left(\left( C_i^r \right)_{i=1}^\K\right)_{r=1}^\R$ such that $\forall i\in[\K], r\in[\R]$, $C_i^r$ is a quantum circuit composed of a Clifford unitary followed by a standard basis measurement of a subset of the qubits on which it operates. Then, there exists a classical strategy $\tilde{S}$ for $G$ that meets the following conditions:

\begin{enumerate}
    \item $\tilde{S}$ is equivalent to $S$.
    \item The computation performed by the provers executing $\tilde{S}$ throughout the interaction is efficient in terms of the circuit size of the provers that implement $S$ and the number of qubits they use. However, computing the random bits that are shared between the provers before the interaction is only known to be efficient in terms of the same parameters if $\ket{\psi}$ is a stabilizer state. 
\end{enumerate}

\end{lemma}

Following is the proof of Lemma \ref{lemma: multi_round_clifford_unitary_classical}.

\begin{proof}

We begin by analyzing the strategy $S$.\newline 

\textbf{Analysis of the Quantum Strategy $S$}\newline 

Let $(P_1, ..., P_K)$ be provers that execute $S$ with the implementation $C$. At the beginning of the game, before the verifier sends its questions, the state of all qubits held by the provers is $\rho := \ket{\psi}\bra{\psi}$. \newline 

For any $i\in[\K], r\in[\R]$, the circuit $C_{i}^r$ operates on the input bit-string received by $P_i$ from the verifier in the $r^{th}$ round of the interaction, and on $P_i$'s part of the shared state at that point. $C_i^r$ is composed of a Clifford unitary $U_{i}^r$ and followed by a standard basis measurement of a subset of the qubits held by the $i^{th}$ prover in the $r^{th}$ round, denoted by $\Lambda_{i}^r$ . $C_i^r$ is equivalent to the application of the elements in the above sequence in chronological order going from left to right.\newline 

For convenience and without loss of generality, we make several assumptions:

\begin{enumerate}
    \item $\exists s,t\in \mathbb{N}: Q\subseteq \{0,1\}^s, A\subseteq \{0,1\}^t$.
    \item Each prover returns the result of the classical computation it performs at the end of each round as its answer to the verifier. We denoted this result by $a_i^r$ for the $i^{th}$ prover in the $r^{th}$ round. 
    \item At each round, any prover stores its received query in qubits that are each initialized to \( \ket{0} \) and are not used by the prover up to that point. This implies that \( U_{i}^{r+1} \) may operate on more qubits than \( U_{i}^r \) does. We denote by \( d_{i}^r \) the number of qubits on which \( U_{i}^r \) operates.
    \item As the last step performed by the $i^{th}$ prover in each round, it measures $t$ of its qubits and returns the measurement result as its answer to the verifier.
\end{enumerate}

\textbf{Definition of a Useful Operator} \newline 

Next, let there be $n\in \mathbb{N}$ and two bit-strings $x,y\in\{0,1\}^n$. We define the operator $X_{x\oplus y}$ as follows:
\[
X_{x\oplus y} = \bigotimes_{i=1}^{n} X^{x_i\oplus y_i}
\]
where $X$ is the Pauli-$X$ operator and $X^0=I$, the identity operator of size $2\times 2$. $\oplus$ is the bitwise XOR operation.
By definition of $X_{x\oplus y}$, the following equality holds:
\[
\ket{x} = X_{x\oplus y} \ket{y}
\]

\textbf{An Equivalent Classical Strategy}\newline 

We are now ready to discuss Algorithm~\ref{alg:classical_simulation_multi_round}, which describes a classical strategy for $G$ that is equivalent to $S$.

\begin{algorithm}[htbp]
\caption{Classical Simulation of a Clifford Strategy}
\label{alg:classical_simulation_multi_round}
\begin{algorithmic}[1]
\STATE 
An arbitrary sequence of questions for all rounds, $q = (q^1, \dots, q^R)$, is hard-coded into the algorithm. The provers sample answers $a^r\sim \Pr_S\left(a^r \mid q^{\leq r}, a^{<r}\right)$, starting from $r=1$ and ascending until $r=R$. Here, $q^{\leq r} = (q^1, \dots, q^r)$ and $a^{< r} = (a^1, \dots, a^{r-1})$, where $q^r: =\left(q_{i}^r\right)_{i=1}^\K$ and $a^r: =\left(a_{i}^r\right)_{i=1}^\K$. $q_{i}^r$ and $a_i^r$ are the question and answer received and sent by the $i^{th}$ prover in the $r^{th}$ round of the sampled interaction. The random bits shared by the provers encode the interaction $h = (q^r, a^r)_{r=1}^\R$. 

\STATE 
Let us denote the interaction of the provers with the verifier by $\tilde{h} = (\tilde{q}^r, \tilde{a}^r)_{r=1}^\R$. At each round $r = 1$ to $R$, each prover $i$ performs the following steps:

\begin{enumerate}[label=\alph*)]
    \item Computes the matrix $X_{q_i^r \oplus \tilde{q}_i^r}$.
    \item Computes the operator $R_{i}^r := U_{i}^r \left( X_{q_i^r\oplus \tilde{q}_i^r} \otimes R_{i}^{r-1}\right) U_{i}^{r\dagger}$.
    \item Computes the bit-string $\tilde{a}_i^r$ such that $\ket{\tilde{a}_{i}^r}\bra{\tilde{a}_{i}^r}\otimes \hat{I}_i^r = R_{i}^r \left( \ket{a_{i}^r}\bra{a_{i}^r} \otimes \hat{I}_{i}^r \right) R_{i}^{r\dagger}$.
    \item Returns $\tilde{a}_i^r$ as an answer to the verifier.
\end{enumerate}

where $R_i^0:=I_{i}$ is the identity operator which acts on the same qubits on which $U_i^1$ acts, excluding those that store the question $\tilde{q}_i^1$. $\hat{I}_{i}^r$ is the identity operator which acts on the same qubits on which $R_{i}^r$ does, excluding those that the $i^{th}$ prover measures to obtain $a_{i}^r$.

\end{algorithmic}
\end{algorithm}

 For a strategy $S$, $\Pr_{S}:H^R\rightarrow [0,1]$ denotes the probability mass function marking the distribution of questions and answers in an interaction between provers that execute $S$ and a verifier which follows the protocol $G$. We denote by $\tilde{S}$ the strategy of provers executing Algorithm \ref{alg:classical_simulation_multi_round} with respect to a quantum strategy $S$. We prove that $S,\tilde{S}$ are equivalent, a condition expressed by the following equality:

\begin{equation}
\begin{aligned}
\forall h \in H^R: \Pr_{\tilde{S}}[h]= \Pr_{S} [h]
\end{aligned}
\end{equation}

Let us begin by proving the following claim regarding the operators $\left(\left(R_i^r\right)_{i=1}^\K\right)_{r=1}^\R$ that are calculated by the provers throughout the execution of the algorithm:

\[\forall i\in[\K], r\in[\R]:R_{i}^r\in \mathcal{P}_{d_{i}^r}\]

where ${\mathcal{P}}_{d_{i}^r}$ is the Pauli group over $d_{i}^r$ qubits—the number of qubits on which $U_i^r$ operates. As mentioned above, $U_{i}^r$ is a unitary composed of Clifford gates. Hence $U_{i}^r\in \mathcal{C}_{d_{i}^r}$, the \( d_{i}^r \)-qubit Clifford group. Furthermore, $X_{q_i^r \oplus \tilde{q}_i^r}\in \mathcal{P}_{s}$, the \( s \)-qubit Pauli group.\newline 

According to the definition of the Clifford group, for any $n\in \mathbb{N}$, a matrix in $\mathcal{C}_{n}$ maps any matrix in the Pauli group $P_{n}$ to another matrix in that group by conjugation. It follows that $\forall i\in[\K]: R_{i}^1\in \mathcal{P}_{d_{i}^1}$. This serves as the basis for an inductive argument that proves the above claim. \newline 

Let there be $r\in [R-1], i\in[\K]$. We assume that $R_{i}^r\in P_{d_{i}^r}$. Hence, $\xi:=X_{q_i^{r+1} \oplus \tilde{q}_i^{r+1}}\otimes R_{i}^{r}\in \mathcal{P}_{d_{i}^{r+1}}$. Therefore, $R_{i}^{r+1}=U_{i}^{r+1}\xi(U_{i}^{r+1})^\dagger \in \mathcal{P}_{d_{i}^{r+1}}$. This completes the inductive argument and proves the claim. \newline 

Next, consider a possible interaction $h = \left( \left( q_i^r, a_i^r \right)_{i=1}^\K \right)_{r=1}^\R$ between a verifier $V$ which follows protocol $G$ and provers $(P_1, \dots, P_K)$ that execute strategy $S$. Let there be another possible series of queries 
\[\tilde{q} = \left( \left( \tilde{q}_i^r \right)_{i=1}^\K \right)_{r=1}^\R\]

which correspond to a full interaction of $V$ with provers. We analyze an interaction between $V$ and classical provers executing Algorithm \ref{alg:classical_simulation_multi_round}, assuming their shared randomness encodes the interaction $h$, while the questions presented to them by the verifier throughout the interaction are $\tilde{q}$. In this analysis, we refer to the values computed during the interaction as $\left( \left( \tilde{a}_{i}^r \right)_{i=1}^\K \right)_{r=1}^\R$.

Hence, we may denote by $\tilde{h}:=\left( \left( \tilde{q}_i^r, \tilde{a}_i^r \right)_{i=1}^\K \right)_{r=1}^\R$ the interaction between the provers which execute Algorithm \ref{alg:classical_simulation_multi_round} when their shared randomness encodes the interaction $h$. We denote the operators that are computed by the algorithm during the interaction and which depend on the questions $q,\tilde{q}$, by
\[
\left( \left( R_{i}^r \right)_{i=1}^\K\right)_{r=1}^\R,
\]

For all $r\in[\R]$, denote by $\rho_{q,a,r}$ the post-measurement state of all qubits held by the $\K$ provers executing the strategy $S$ right after returning the answers of the $r^{th}$ round to the verifier, assuming the following conditions hold:

\begin{itemize}
    \item For each $j \in [r]$, the provers receive the questions $q^j=\left(q_i^j\right)_{i=1}^\K$ from the verifier in the $j^{th}$ round.
    \item For each $j \in [r]$, the final answers returned by the provers in the $j^{th}$ round are $a^j = \left( a_i^j \right)_{i=1}^\K$.   
\end{itemize}

We begin by examining $\rho_{\tilde{q}, \tilde{a}, 1}$:
\begin{align*}
\rho_{\tilde{q}, \tilde{a}, 1} &= \left( \bigotimes_{i=1}^\K \left( \ket{\tilde{a}_{i}^1}\bra{\tilde{a}_{i}^1} \otimes \hat{I}_{i}^1 \right) \right) \left( \bigotimes_{i=1}^\K U_{i}^1 \right) \left(( \bigotimes_{i=1}^\K \ket{\tilde{q}_i^1}\bra{\tilde{q}_i^1})\otimes\rho \right)  \left( \bigotimes_{i=1}^\K U_{i}^{1\dagger} \right) \left( \bigotimes_{i=1}^\K \left( \ket{\tilde{a}_{i}^1}\bra{\tilde{a}_{i}^1} \otimes \hat{I}_{i}^1 \right) \right)
\end{align*}

Observe the following equality:

\[
\ket{\tilde{a}_{i}^1}\bra{\tilde{a}_{i}^1}\otimes \hat{I}_{i}^1 =  R_{i}^1 \left( \ket{a_{i}^1}\bra{a_{i}^1} \otimes \hat{I}_{i}^1\right)(R_{i}^{1})^\dagger,
\]

Using it, we can express $\rho_{\tilde{q}, \tilde{a}, 1}$ as:
\begin{align*}
\rho_{\tilde{q}, \tilde{a}, 1} &= \left( \bigotimes_{i=1}^\K R_{i}^1 \left( \ket{a_{i}^1}\bra{a_{i}^1} \otimes \hat{I}_{i}^1 \right) R_{i}^{1\dagger} \right) \left( \bigotimes_{i=1}^\K U_{i}^1 \right)\left(( \bigotimes_{i=1}^\K \ket{\tilde{q}_i^1}\bra{\tilde{q}_i^1})\otimes\rho \right) \\
&\quad\left( \bigotimes_{i=1}^\K U_{i}^{1\dagger} \right)  \left( \bigotimes_{i=1}^\K R_{i}^1 \left( \ket{a_{i}^1}\bra{a_{i}^1} \otimes \hat{I}_{i}^1 \right) (R_{i}^{1})^\dagger \right) 
\end{align*}

Since $(R_{i}^1)^\dagger=R_{i}^1$, then $(R_{i}^1)^\dagger U_{i}^1= (R_{i}^1) U_{i}^1= U_{i}^1(X_{q_i^1 \oplus \tilde{q}_i^1}\otimes I_{i}^1) (U_{i}^1)^\dagger U_{i}^1=U_{i}^1(X_{q_i^1 \oplus \tilde{q}_i^1}\otimes I_{i}^1)$. As a result, 

\begin{align*}
\rho_{\tilde{q}, \tilde{a}, 1} &= \left( \bigotimes_{i=1}^\K R_{i}^1 \left( \ket{a_{i}^1}\bra{a_{i}^1} \otimes \hat{I}_{i}^1 \right)U_{i}^1(X_{q_i^1 \oplus \tilde{q}_i^1}\otimes I_{i}^1)\right)\left(( \bigotimes_{i=1}^\K \ket{\tilde{q}_i^1}\bra{\tilde{q}_i^1})\otimes\rho \right) \\
&\quad\left( \bigotimes_{i=1}^\K R_{i}^1 \left( \ket{a_{i}^1}\bra{a_{i}^1} \otimes \hat{I}_{i}^1 \right)U_{i}^1(X_{q_i^1 \oplus \tilde{q}_i^1}\otimes I_{i}^1)\right)^\dagger
\end{align*}

$X_{q_i^1 \oplus \tilde{q}_i^1}\ket{\tilde{q}_i^1}\bra{\tilde{q}_i^1}X_{q_i^1 \oplus \tilde{q}_i^1}=\ket{q_i^1}\bra{q_i^1}$, and therefore:
\begin{align*}
\rho_{\tilde{q}, \tilde{a}, 1} &= \left( \bigotimes_{i=1}^\K R_{i}^1\right) \left( \bigotimes_{i=1}^\K\ket{a_{i}^1}\bra{a_{i}^1} \otimes \hat{I}_{i}^1 \right) \left(\bigotimes_{i=1}^\K U_{i}^1\right)\left(( \bigotimes_{i=1}^\K \ket{q_i^1}\bra{q_i^1})\otimes\rho \right) \\
&\quad \left(\bigotimes_{i=1}^\K U_{i}^1\right)^\dagger \left( \bigotimes_{i=1}^\K\ket{a_{i}^1}\bra{a_{i}^1} \otimes \hat{I}_{i}^1 \right) \left( \bigotimes_{i=1}^\K R_{i}^1\right)^\dagger = \left( \bigotimes_{i=1}^\K R_{i}^1\right)\rho_{q, a, 1}\left( \bigotimes_{i=1}^\K R_{i}^1\right)^\dagger 
\end{align*}

This serves as the basis for an inductive argument over \( r \) proving the following statement:

\begin{equation}
\label{eq:state_link}
\forall r\in[\R]: \quad \rho_{\tilde{q}, \tilde{a}, r} = \left( \bigotimes_{i=1}^\K R_{i}^r \right) \rho_{q, a, r} \left( \bigotimes_{i=1}^\K R_{i}^r \right)^\dagger
\end{equation}

Let \( 2\leq r\leq R \). As the induction hypothesis, assume that the following holds:

\begin{equation}
\label{eq:state_link_induction}
\quad \rho_{\tilde{q}, \tilde{a}, r-1} = \left( \bigotimes_{i=1}^\K R_{i}^{r-1} \right) \rho_{q, a, r-1} \left( \bigotimes_{i=1}^\K R_{i}^{r-1} \right)^\dagger
\end{equation}

We observe the following equality:

\begin{align}
\rho_{\tilde{q}, \tilde{a}, r} &= \left( \bigotimes_{i=1}^\K \left( \ket{\tilde{a}_{i}^r}\bra{\tilde{a}_{i}^r} \otimes \hat{I}_{i}^r \right) \right) \left( \bigotimes_{i=1}^\K U_{i}^r \right) \left( \left( \bigotimes_{i=1}^\K \ket{\tilde{q}_i^r}\bra{\tilde{q}_i^r} \right) \otimes \rho_{\tilde{q}, \tilde{a}, r-1} \right) \left( \bigotimes_{i=1}^\K U_{i}^r \right)^\dagger \notag \\
&\quad \left( \bigotimes_{i=1}^\K \left( \ket{\tilde{a}_{i}^r}\bra{\tilde{a}_{i}^r} \otimes \hat{I}_{i}^r \right) \right)
\end{align}

Using the induction hypothesis, the relation \( (\ket{\tilde{a}_{i}^r}\bra{\tilde{a}_{i}^r} \otimes \hat{I}_{i}^r) = R_{i}^r (\ket{a_{i}^r}\bra{a_{i}^r} \otimes \hat{I}_{i}^r))R_{i}^{r\dagger} \) and the fact \( (R_{i}^r)^\dagger = R_{i}^r \), we obtain the following equality:

\begin{align}
\rho_{\tilde{q}, \tilde{a}, r} &= \left( \bigotimes_{i=1}^\K R_{i}^r \left( \ket{a_{i}^r}\bra{a_{i}^r} \otimes \hat{I}_{i}^r \right) R_{i}^r \right)\notag \\
&\quad  \left( \bigotimes_{i=1}^\K U_{i}^r \right) \left( \left( \bigotimes_{i=1}^\K \ket{\tilde{q}_i^r}\bra{\tilde{q}_i^r} \right) \otimes \left(\left( \bigotimes_{i=1}^\K R_{i}^{r-1} \right) \rho_{q,a,r-1} \left( \bigotimes_{i=1}^\K R_{i}^{r-1} \right)^\dagger \right)\right) \notag \\
&\quad \left( \bigotimes_{i=1}^\K U_{i}^r \right)^\dagger \left( \bigotimes_{i=1}^\K R_{i}^r \left( \ket{a_{i}^r}\bra{a_{i}^r} \otimes \hat{I}_{i}^r \right) R_{i}^r \right)
\end{align}

Since $R_{i}^r=U_{i}^{r} (X_{q_i^r \oplus \tilde{q}_i^r}\otimes R_{i}^{r-1}) (U_{i}^{r})^{\dagger}$, then $R_{i}^rU_{i}^r=U_{i}^r(X_{q_i^r \oplus \tilde{q}_i^r}\otimes R_{i}^{r-1})$. As a result, we have:
\begin{align}
\rho_{\tilde{q}, \tilde{a}, r} &= \left( \bigotimes_{i=1}^\K R_{i}^r \left( \ket{a_{i}^r}\bra{a_{i}^r} \otimes \hat{I}_{i}^r \right) U_{i}^r(X_{q_i^r \oplus \tilde{q}_i^r}\otimes R_{i}^{r-1}) \right) \notag \\
&\left( \left( \bigotimes_{i=1}^\K \ket{\tilde{q}_i^r}\bra{\tilde{q}_i^r} \right) \otimes \left(\left( \bigotimes_{i=1}^\K R_{i}^{r-1} \right) \rho_{q,a,r-1} \left( \bigotimes_{i=1}^\K R_{i}^{r-1} \right)^\dagger \right)\right) \notag \\
&\quad  \left( \bigotimes_{i=1}^\K R_{i}^r \left( \ket{a_{i}^r}\bra{a_{i}^r} \otimes \hat{I}_{i}^r \right) U_{i}^r(X_{q_i^r \oplus \tilde{q}_i^r}\otimes R_{i}^{r-1}) \right)^\dagger
\end{align}

We use the fact that $R_{i}^{r-1}$ is in the Pauli group and the fact that $X_{q_i^r \oplus \tilde{q}_i^r}\ket{\tilde{q}_i^r}=\ket{q_i^r}$, to derive the following result:
\begin{align}
\rho_{\tilde{q}, \tilde{a}, r} &= \left( \bigotimes_{i=1}^\K R_{i}^r \right) \left(\bigotimes_{i=1}^\K \left(\ket{a_{i}^r}\bra{a_{i}^r} \otimes \hat{I}_{i}^r \right)\right) \left( \bigotimes_{i=1}^\K U_{i}^r\right) \left( \left( \bigotimes_{i=1}^\K \ket{q_i^r}\bra{q_i^r} \right) \otimes  \rho_{q,a,r-1} \right)\notag \\
&\quad  \left( \bigotimes_{i=1}^\K U_{i}^r\right)^\dagger \left(\bigotimes_{i=1}^\K \left(\ket{a_{i}^r}\bra{a_{i}^r} \otimes \hat{I}_{i}^r \right)\right) \left( \bigotimes_{i=1}^\K R_{i}^r \right)^\dagger = \left( \bigotimes_{i=1}^\K R_{i}^r \right) \rho_{q,a,r} \left( \bigotimes_{i=1}^\K R_{i}^r \right)^\dagger
\end{align}

This completes the induction step over \( r \). Thus, we have proven statement~\eqref{eq:state_link}. Based on it, we prove the following equality:

\begin{align}
\label{eq:prob_eq}
&\forall r\in[\R]: \Pr_{S}\left[\tilde{a}^r \mid \tilde{h}^{r-1}, \tilde{q}^r\right]=\Pr_{S}\left[a^r\mid h^{r-1}, q^r\right] 
\end{align}

where $h^{r-1}:=(q^1,a^1,...,q^{r-1},a^{r-1})$, $\tilde{h}^{r-1}:=(\tilde{q}^1,\tilde{a}^1,...,\tilde{q}^{r-1},\tilde{a}^{r-1})$. \newline

Let there be $r\in[\R]$. 
\begin{align}
&\Pr_{S}\left[\tilde{a}^r\mid\tilde{h}^{r-1}, \tilde{q}^r\right]=Tr[\left(\ket{\tilde{a}^r}\bra{\tilde{a}^r}\otimes \hat{I}^r\right)\left(\bigotimes_{i=1}^\K U_{i}^r\right)\left(\ket{\tilde{q}^r}\bra{\tilde{q}^r}\otimes\rho_{\tilde{q}, \tilde{a}, r-1}\right)\left(\bigotimes_{i=1}^\K U_{i}^r\right)^{\dagger} ] 
\end{align}

We use the equality $\rho_{\tilde{q}, \tilde{a}, r-1}=\left(\bigotimes_{i=1}^\K R_{i}^{r-1}\right)\rho_{q, a, r-1}\left(\bigotimes_{i=1}^\K R_{i}^{r-1}\right)^\dagger$, to obtain the following result:

\begin{align}
&\Pr_{S}\left[\tilde{a}^r\mid \tilde{h}^{r-1}, \tilde{q}^r\right]=Tr[\left(\ket{\tilde{a}^r}\bra{\tilde{a}^r}\otimes \hat{I}^r\right)\left(\bigotimes_{i=1}^\K U_{i}^r\right)\notag \\
&\left(\ket{\tilde{q}^r}\bra{\tilde{q}^r}\otimes\left(\bigotimes_{i=1}^\K R_{i}^{r-1}\right)\rho_{q, a, r-1}\left(\bigotimes_{i=1}^\K R_{i}^{r-1}\right)^\dagger\right)\left(\bigotimes_{i=1}^\K U_{i}^r\right)^{\dagger} ]
\end{align}

Next, observe that $X_{q_i^r \oplus \tilde{q}_i^r}\ket{\tilde{q}_{i}^r}=\ket{q_{i}^r}$. With that observation in hand, we continue as follows:

\begin{align}
&\Pr_{S}\left[\tilde{a}^r\mid \tilde{h}^{r-1}, \tilde{q}^r\right] = Tr[\left(\ket{\tilde{a}^r}\bra{\tilde{a}^r}\otimes \hat{I}^r\right)\left(\bigotimes_{i=1}^\K U_{i}^r\left( X_{q_i^r \oplus \tilde{q}_i^r}\otimes R_{i}^{r-1}\right)\right)\notag \\
&\left( \ket{q^r}\bra{q^r}\otimes \rho_{q, a, r-1}\right) \left(\bigotimes_{i=1}^\K U_{i}^r\left(X_{q_i^r \oplus \tilde{q}_i^r}\otimes R_{i}^{r-1}\right)\right)^\dagger]
\end{align}

Since $R_{i}^r=U_{i}^r \left( X_{q_i^r \oplus \tilde{q}_i^r}\otimes R_{i}^{r-1}\right)\left(U_{i}^r\right)^\dagger$, then $U_{i}^r\left( X_{q_i^r \oplus \tilde{q}_i^r}\otimes R_{i}^{r-1}\right)=R_{i}^r U_{i}^r$. Then, we derive the following equalities:

\begin{align}
&\Pr_{S}\left[\tilde{a}^r\mid \tilde{h}^{r-1}, \tilde{q}^r\right] \notag \\
&= Tr[\left(\ket{\tilde{a}^r}\bra{\tilde{a}^r}\otimes \hat{I}^r\right)\left(\bigotimes_{i=1}^\K R_{i}^rU_{i}^r\right)\left( \ket{q^r}\bra{q^r}\otimes \rho_{q, a, r-1}\right) \left(\bigotimes_{i=1}^\K R_{i}^rU_{i}^r\right)^\dagger]\notag \\
&= Tr[\left(\ket{\tilde{a}^r}\bra{\tilde{a}^r}\otimes \hat{I}^r\right)\left(\bigotimes_{i=1}^\K R_{i}^r\right)\left(\bigotimes_{i=1}^\K U_{i}^r\right)\left( \ket{q^r}\bra{q^r}\otimes \rho_{q, a, r-1}\right) \left(\bigotimes_{i=1}^\K U_{i}^r\right)^\dagger\left(\bigotimes_{i=1}^\K R_{i}^r\right)^\dagger]\notag \\
&= Tr[\left(\bigotimes_{i=1}^\K R_{i}^r\right)^\dagger\left(\ket{\tilde{a}^r}\bra{\tilde{a}^r}\otimes \hat{I}^r\right)\left(\bigotimes_{i=1}^\K R_{i}^r\right)\left(\bigotimes_{i=1}^\K U_{i}^r\right)\left( \ket{q^r}\bra{q^r}\otimes \rho_{q, a, r-1}\right) \left(\bigotimes_{i=1}^\K U_{i}^r\right)^\dagger] \notag \\
&= Tr[\left(\bigotimes_{i=1}^\K R_{i}^r\right)\left(\ket{\tilde{a}^r}\bra{\tilde{a}^r}\otimes \hat{I}^r\right)\left(\bigotimes_{i=1}^\K R_{i}^r\right)^\dagger\left(\bigotimes_{i=1}^\K U_{i}^r\right)\left( \ket{q^r}\bra{q^r}\otimes \rho_{q, a, r-1}\right) \left(\bigotimes_{i=1}^\K U_{i}^r\right)^\dagger] \notag \\
&= Tr[\left(\ket{a^r}\bra{a^r}\otimes \hat{I}^r\right)\left(\bigotimes_{i=1}^\K U_{i}^r\right)\left( \ket{q^r}\bra{q^r}\otimes \rho_{q, a, r-1}\right)\left(\bigotimes_{i=1}^\K U_{i}^r\right)^\dagger] = \Pr_{S}\left[a^r\mid h^{r-1}, q^r\right]
\end{align}

where in the second-to-last equality, we rely on the fact that 
\[\left(\bigotimes_{i=1}^\K R_{i}^r\right)\left(\ket{\tilde{a}^r}\bra{\tilde{a}^r}\otimes \hat{I}^r\right)\left(\bigotimes_{i=1}^\K R_{i}^r\right)^\dagger=\left(\ket{a^r}\bra{a^r}\otimes \hat{I}^r\right)\] 

Thus, we have proven $\eqref{eq:prob_eq}$. Next, we observe that if the interaction of the provers which execute Algorithm \ref{alg:classical_simulation_multi_round} in the first $r-1$ rounds is $\tilde{h}^{r-1}$, and the questions asked in the $r^{th}$ round by the verifier are $\tilde{q}^r:=(\tilde{q}_1^r,...,\tilde{q}_K^r$), then the answers returned by the $i^{th}$ prover in the $r^{th}$ round are:
\[R_{i}^r(\ket{a_{i}^{r}}\bra{a_{i}^{r}}\otimes I_i^r)(R_{i}^r)^{
\dagger}\]

where $a_{i}^{r}$ is the answer of the $i^{th}$ prover in the $r^{th}$ round of the interaction that is encoded in the random bits shared by the provers. It follows that $\Pr_{\tilde{S}}\left[\tilde{a}^{r} \mid \tilde{h}^{r-1}, \tilde{q}^r\right]=\Pr_{S}\left[ a^r \mid h^{r-1}, q^r\right]$, where $a^r:=(a_1^r, ...,a_K^r)$ such that for any $i\in[\K]$, $\ket{\tilde{a}_{i}^{r}}\bra{\tilde{a}_{i}^{r}}\otimes I_i^r := R_{i}^r(\ket{a_{i}^{r}}\bra{a_{i}^{r}}\otimes I_i^r)(R_{i}^r)^{
\dagger}$. Recall that the probability of the verifier asking a question $q^r:=(q_1,...,q_K)$ in the $r^{th}$ round assuming the interaction history $h^{r-1}$ is denoted by $\pi(q^r\mid h^{r-1})$.
Then,

\begin{align*}
&\Pr_{\tilde{S}}\left[\tilde{h}^R\right] = \prod_{r=1}^\R \pi(\tilde{q}^r|\tilde{h}^{r-1})\cdot \Pr_{\tilde{S}}\left[\tilde{a}^{r} \mid \tilde{h}^{r-1}, \tilde{q}^r\right] = \prod_{r=1}^\R \pi(\tilde{q}^r|\tilde{h}^{r-1})\cdot \Pr_{S}\left[a^r \mid h^{r-1}, q^r\right ] \notag \\
& = \prod_{r=1}^\R \pi(\tilde{q}^r|\tilde{h}^{r-1})\cdot \Pr_{S}\left[\tilde{a}^r \mid \tilde{h}^{r-1}, \tilde{q}^r\right] = \Pr_{S}(\tilde{h}^R)
\end{align*}

where the second to last equality is due to Eq.\ (\ref{eq:prob_eq}). This proves that $S,\tilde{S}$ are equivalent strategies, so Algorithm \ref{alg:classical_simulation_multi_round} is correct.\newline

\textbf{Running Time Complexity Analysis} \newline 

We now turn to discuss the running time complexity of Algorithm \ref{alg:classical_simulation_multi_round}. We begin with step 1, in which the provers draw their shared random bits. If provers executing $S$ share a stabilizer state, then the random bits can be drawn by simulating the provers classically using the Gottesman-Knill Theorem. The operation of different provers in the same round can be computed in parallel, and consequently
the time complexity of classical simulation is $O(\text{poly}(d + \sum_{r=1}^\R \max_{i=1,...,\K}\{|C_i^r|\}))$, where $|C_i^r|$ is the number of Clifford gates and classical logical gates in the circuit, and $d$ is the number of qubits held by all provers together. However, if we do not assume the provers share a stabilizer state, we do not provide a non-trivial bound on the running time complexity required to complete step 1. \newline 

Nevertheless, the next steps, which describe the computation of the classical provers during the interaction with the verifier, can be computed efficiently. Steps 2 and 3 are completed using the method for calculating the mapping of one matrix in the Pauli group to another by conjugation with a Clifford matrix, which is described in the Preliminaries section (\ref{explanation:clifford_pauli_update}). Here, too, the operation of different provers in the same round can be computed in parallel.  $d^r:=\max_{i=1,...,\K}\{d_{i,\ell^r}^r\}$ is the maximal number of qubits held by a prover in the $r^{th}$ round of the interaction. The time complexity of steps 2 and 3 together is as follows: $O(\sum_{r=1}^\R \max\{d^r, \max_{i=1,...,\K}\{|C_i^r|\}\})$. This is the overall running time complexity of the computation of the provers during the interaction, because step 4 takes time $O(\sum_{r=1}^\R d^r)$.
\end{proof}

We are now ready to prove Theorem \ref{thm: multi_round_clifford_classical}.

\begin{proof}
In the proof, we construct an interactive protocol $G_S$ based on the protocol $G$ and the strategy $S$, where the classical computations performed by $S$ on measured qubits are delegated to the verifier. To describe the game, we first assume without loss of generality that in any round of communication of $G$, each of the provers which execute $S$ performs $\ell$ measurements. After each measurement, a prover may perform a classical computation on the measurement result or not. The game $G_S$ is defined as follows:

\begin{enumerate}
    \item The verifier samples the first round's questions from the same distribution used in $G$, and sends the questions to the provers.
    \item The honest provers follow the strategy $S$, only instead of conducting classical post-processing of measurement results, they send their measurement results to the verifier. The verifier performs the appropriate classical computations for each of the provers (according to the strategy $S$), and sends each of them its corresponding computation result, which the prover then uses to continue its computation. If a prover only needs to measure qubits without performing a classical computation on the result, the verifier simply sends back to the prover the measurement result unchanged. 
    This adds $\ell$ rounds of communication in $G_S$ for each round of communication in $G$. 
    \item After all provers completed their computation which corresponds to a certain round $r\in[\R]$ in $G$, they send their answers the verifier. The verifier then sends them the next round of questions if $r<R$, and otherwise decides whether to accept or reject.
    \item At each round, when the verifier samples the next round's questions or decides whether to accept or to reject, it acts exactly as the verifier in $G$ does. That is, the verifier ignores the rounds of communication in which classical computations were delegated to it by the provers, and only considers rounds in which the provers sent values that correspond to answers to its queries.
\end{enumerate}

We denote the honest quantum strategy in the game $G_S$ by $S'$. This strategy is consistent with the strategy $S$ in the game $G$, except that in $S'$ the provers delegate their classical post-processing of measurement results to the verifier. The computation of provers executing $S'$ at each round of $G_S$ can clearly be implemented using a Clifford unitary followed by a standard basis measurement of a subset of the qubits on which the unitary acts. This is because in $S'$, the provers only perform a measurement at the end of each round's computation, and delegate their classical computation to the verifier. Therefore, according to Lemma~$\ref{lemma: multi_round_clifford_unitary_classical}$, there exists a classical strategy $\tilde{S'}$ for $G_S$ which is equivalent to the strategy $S'$ in $G_S$. \newline 

Using $\tilde{S'}$, we can naturally define a strategy $\tilde{S}$ for $G$, which is equivalent to the strategy $S$. Provers implementing $\tilde{S}$ act like those implementing $\tilde{S'}$, except rather than delegating classical computations to the verifier, they perform all the computations by themselves. The fact that $\tilde{S}$ is equivalent to $S$ stems directly from the fact that $\tilde{S'}$ is equivalent to $S'$ and the definition of the game $G_S$.

\end{proof}

\subsection{ \CliffordMIP$= \cc{MIP}$}

In this section, we state the following equality between complexity classes. It is a straightforward corollary of Theorem \ref{thm: multi_round_clifford_classical}, but we provide a full proof of it for the sake of completeness.

\begin{theorem}
\CliffordMIP $= \cc{MIP}$ 
\end{theorem}

\begin{proof} We will show that \CliffordMIP$= \cc{MIP}$ by proving both inclusions. \newline 

\textbf{(1) \CliffordMIP $ \subseteq \cc{MIP}$}

Let $L \in$ \CliffordMIP. Then $L$ is decided by an interactive proof system $G$ where a classical verifier interacts with quantum provers limited to Clifford gates and classical post-processing of measurement results. By Theorem~\ref{thm: multi_round_clifford_classical}, any such quantum strategy $S$ can be replaced by an equivalent classical strategy $S^\ast$ without loss of performance.\newline 

\emph{Completeness:} If $x \in L$, the quantum provers have a strategy $S$ that convinces the verifier to accept with probability at least $c(|x|)$. The equivalent classical strategy $S^\ast$ achieves the same acceptance probability, satisfying the completeness condition of $\cc{MIP}$.

\emph{Soundness:} If $x \notin L$, no strategy using Clifford gates and classical post-processing of measurement results can convince the verifier to accept with probability greater than $s(|x|)$, hence no classical strategy can do it as well. \newline 

Thus, $L \in \cc{MIP}$. \newline 

\textbf{(2) $\cc{MIP} \subseteq$ \CliffordMIP}

Let $L \in \cc{MIP}$. Then there exists an interactive proof system where a classical verifier interacts with classical provers deciding $L$ with completeness $c(|x|)$ and soundness $s(|x|)$. We show that this protocol has the same completeness and soundness parameters in the \CliffordMIP setting.\newline 

\emph{Completeness:} If $x \in L$, there exist classical provers that convince the verifier to accept with probability at least $c(|x|)$. The \CliffordMIP provers replicate this strategy using the same classical operations, thereby achieving the same acceptance probability.

\emph{Soundness:} If $x \notin L$, no classical strategy can convince the verifier to accept with probability greater than $s(|x|)$. Hence, if there existed a strategy in the \CliffordMIP setting surpassing this soundness parameter, it would contradict Theorem \ref{thm: multi_round_clifford_classical}. \newline

Thus, $L \in$ \CliffordMIP. 

\end{proof}

\section{Non-Local Games with At Most One Non-Clifford Prover}
In this section, we address a question posed by Kalai et al.\ \cite{kalai2023quantum} (for details, see the Introduction section) by proving the following theorem.

\begin{theorem}
\label{thm:non_local_games}
Let \( G = \left( Q, A, \pi ,v\right) \) be a non-local game with \( K \) provers, and let \( S = \left( \ket{\psi},\ (M_i)_{i=1}^\K \right) \) be a quantum strategy for \( G \). Suppose there exists an implementation of \( S \), denoted \( (C_i)_{i=1}^\K \), such that for all but at most one index $i\in[\K]$, the circuit \( C_i \) consists only of Clifford gates and classical post-processing of measurement results. Then, there exists a classical strategy for \( G \) that is equivalent to \( S \).
\end{theorem}

From Theorem \ref{thm:non_local_games}, we immediately obtain the following corollary.

\begin{corollary}
Let there be a non-local game $G = \left( Q,A, \pi,v \right)$ with $K$ provers, and a quantum strategy $S = \left( \ket{\psi},\ \left( M_i\right)_{i=1}^\K\right)$ for $G$. If the success probability of $S$ in $G$ exceeds that of any classical strategy, then for any implementation $(C_i)_{i=1}^\K$ of $S$, at most $K-2$ of the circuits $C_i$ ($i\in[\K]$) consist only of Clifford gates in addition to classical post-processing of measurement results.
\end{corollary}

We now turn to present the proof of Theorem \ref{thm:non_local_games}.
\begin{proof}
Without loss of generality, assume that the circuits \( C_1, C_2, \dots, C_{\K-1} \) consist only of Clifford gates and classical post-processing of measurement results; that is, the \( \K \)-th prover may perform general quantum operations. \newline 

Consider the interaction between the verifier \( V \) of the game \( G \) and $\K$ provers executing the strategy \( S \). By Theorem~\ref{thm: multi_round_clifford_classical}, there exists a classical strategy \( \tilde{S} \) for \( \K-1 \) provers, such that when these $\K-1$ provers interact with $V$ using $\tilde{S}$, the resulting distribution of questions and answers is identical to that produced by the first $\K-1$ provers executing strategy $S$ with $V$. \newline 

An interaction of \( \K-1 \) provers with \( V \) is defined as follows:

\begin{enumerate}
    \item \( V \) draws \( \K \) questions \( q = (q_1, q_2, \dots, q_\K) \) at random from the distribution \( \pi \).
    \item It sends the first \( \K-1 \) questions \( q_1, q_2, \dots, q_{\K-1} \) to the corresponding provers and discards the \( K \)-th question \( q_K \).
    \item The \( \K-1 \) provers send their answers \( (a_1, a_2, \dots, a_{\K-1}) \) back to \( V \).
\end{enumerate}

We now define a classical strategy for \( \K \) provers which is equivalent to $S$:

\begin{algorithm}[H]
\caption{Classical Simulation of a Mostly-Clifford Strategy in a Non-Local Game}
\label{alg:classical_simulation_non_local_game}
\begin{algorithmic}[1]
\STATE A fixed arbitrary sequence of questions \( q := (q_1, \dots, q_{\K}) \) is hard-coded into the algorithm. The corresponding answers \( a: = (a_1, \dots, a_\K) \) are sampled from the distribution \( \Pr_{S}[a \mid q] \). The provers share an encoding of the interaction \( h = (q, a) \).
\STATE While interacting with the verifier, the first \( \K-1 \) provers execute the classical strategy \( \tilde{S} \) using Algorithm~\ref{alg:classical_simulation_multi_round} with the hard-coded questions \( q \).
\STATE The \( \K \)-th prover, upon receiving a question \( \tilde{q}_\K \) from \( V \), samples its answer from the conditional distribution \( \Pr_{S}[a_\K \mid q_1,...,q_{\K-1}, \tilde{q}_\K, a_1, \dots, a_{\K-1}]\). This is the distribution of answers of the last out of the $\K$ provers executing the quantum strategy $S$, conditioned on the questions sent to the provers being $q_1,...,q_{\K-1},\tilde{q}_\K$ and on the first $\K-1$ provers returning the answers $a_1,...,a_{\K-1}$.
\end{algorithmic}
\end{algorithm}

Let there be quantum provers denoted \( P_1, \dots, P_\K \) that execute \( S \) with the implementation \( C \). Assume, without loss of generality, that each prover measures all the qubits they hold as the last step of their computation and returns the measurement outcomes as their answer. This assumption does not limit generality, since if only a subset of the measurement results serves as the prover's answer, or if only a subset of the qubits are measured as the final step, then measuring all qubits and sending all outcomes does not change the probability distribution of the measurement results over the relevant subset of qubits. \newline 

We denote by $S'$ the classical strategy of provers executing Algorithm \ref{alg:classical_simulation_non_local_game} with respect to the quantum strategy $S$. For any possible interaction with the verifier \(\tilde{h} = (\tilde{q}, \tilde{a})\), let \( \Pr_{S'}[\tilde{a} \mid \tilde{q}] \) denote the probability that provers which execute the strategy $S'$ return the answers \( \tilde{a}:=(\tilde{a}_1,...,\tilde{a}_\K)\) upon receiving the questions \( \tilde{q}:=(\tilde{q}_1,...,\tilde{q}_\K) \) from the verifier. Similarly, let \(  \Pr_{S}[\tilde{a} \mid \tilde{q}] \) denote the analogous probability for provers that execute the strategy $S$. To prove the correctness of Algorithm \ref{alg:classical_simulation_non_local_game}, we need to show that for any such interaction, \( \Pr_{S'}[\tilde{a} \mid \tilde{q}] = \Pr_{S}[\tilde{a} \mid \tilde{q}] \). \newline 

To simplify our proof, we assume that the computation performed by any of the first $\K-1$ provers in the implementation $C$ of the strategy $S$ consists of applying a Clifford unitary followed by a standard basis measurement. The provers may deviate from this scheme and perform intermediate measurements as well as classical post-processing of measurement results. To extend our proof from the constrained case (no intermediate measurement nor classical computations on measured qubit) which we handle explicitly to the general case, the classical computations performed by the first $\K-1$ provers on measured qubits can be delegated to the verifier—the same approach we took in order to generalize Lemma \ref{lemma: multi_round_clifford_unitary_classical}
to Theorem \ref{thm: multi_round_clifford_classical}. For details, see the proof of Theorem \ref{thm: multi_round_clifford_classical}. \newline 

We can represent $C_i$ as a tuple $(U_{i}, \Lambda_{i})$, where $U_{i}$ is a unitary composed of Clifford gates and $\Lambda_{i}$ is a standard basis measurement of all qubits held by the $i^{th}$ prover.  $C_i$ is equivalent to applying $U_i$ and then $\Lambda_i$. \newline 

We denote the questions which the provers receive from the verifier during the interaction by $\tilde{q}:=(\tilde{q}_1,...,\tilde{q}_\K)$, and the interaction encoded as the shared randomness of the provers by \( h = (q, a) \). Consider the following operators, which are defined in Algorithm~\ref{alg:classical_simulation_multi_round}:

\[\forall i\in[\K]-1: R_{i} := U_{i} (X_{q_i \oplus \tilde{q}_i}\otimes I_{i}) (U_{i})^{\dagger}\]

Since the first $\K-1$ out of $\K$ provers executing algorithm 2 utilize Algorithm~\ref{alg:classical_simulation_multi_round}, their answers are $(\tilde{a}_1, ..., \tilde{a}_{\K-1})$, such that:
\[
\forall i\in[\K]-1: \ket{\tilde{a}_i}\bra{\tilde{a}_i} = R_{i} \ket{a_i}\bra{a_i} R_{i}^\dagger.
\]

We have shown in the proof of Theorem \ref{thm: multi_round_clifford_classical} that for all \( i\in[\K-1] \) \( R_{i} \) is in the Pauli group, and consequently indeed for any $i\in[\K-1]$ there exists a bit-string $\tilde{\alpha}_i$ which meets the above condition. \newline 

Assume, without loss of generality, that the $\K-1$ Clifford provers executing the quantum strategy \( S \) with the implementation \( C \) perform their computation and send their answers back to the verifier before the $\K$-th (non-Clifford) prover performs any opearation on its qubits. Denote by \( \rho_{\tilde{q}, \tilde{a}} \) the shared state of all \( \K \) provers executing $S$ after the first \( \K-1 \) provers send their answers to the verifier, and before the \( \K \)-th prover performs any operation, conditioned on the series of questions sent by the verifier to the provers being \( \tilde{q}:=(\tilde{q}_1,...,\tilde{q}_\K) \), and on the first \( \K-1 \) provers having returned the answers \( \tilde{a} = (\tilde{a}_1, \dots, \tilde{a}_{\K-1}) \). $\rho_{q,a}$ is defined analogously. In the proof of Theorem \ref{thm: multi_round_clifford_classical}, we have shown that the following holds:

\[\rho_{\tilde{q},\tilde{a}} = ((\bigotimes_{i=1}^{\K-1}R_{i})\otimes I) \rho_{q,a} ((\bigotimes_{i=1}^{\K-1}R_{i})\otimes I)^\dagger\]
where $I$ is the identity operator acting on the qubits held by the $\K^{th}$ prover. Since we assumed the first $\K-1$ provers measure all their qubits, then $\rho_{\tilde{q},\tilde{a}}, \rho_{q,a}$ can be represented by $\ket{\tilde{a}}\bra{\tilde{a}}\otimes\sigma_{\K,\tilde{a}}, \ket{a}\bra{a}\otimes\sigma_{\K,a}$, respectively, where  $\sigma_{\K,b}, \sigma_{\K,a}$ are density matrices representing quantum states of the qubits held by the $\K^{th}$ prover. Hence, the following equalities hold:
\begin{align}
& \ket{\tilde{a}}\bra{\tilde{a}}\otimes\sigma_{\K,\tilde{a}} = \rho_{\tilde{q},\tilde{a}} = ((\bigotimes_{i=1}^{\K-1}R_{i})\otimes I) \rho_{q,a} ((\bigotimes_{i=1}^{\K-1}R_{i})\otimes I)^\dagger \notag \\
& = ((\bigotimes_{i=1}^{\K-1}R_{i})\otimes I) (\ket{a}\bra{a}\otimes\sigma_{\K,a}) ((\bigotimes_{i=1}^{\K-1}R_{i})\otimes I)^\dagger = (\bigotimes_{i=1}^{\K-1}R_{i})(\ket{a}\bra{a})(\bigotimes_{i=1}^{\K-1}R_{i})^\dagger \otimes \sigma_{\K,a} =  \ket{\tilde{a}}\bra{\tilde{a}}\otimes\sigma_{\K,a}
\end{align}

It follows that 
\begin{align}
\label{eq:res_state_equality_1}
\sigma_{\K,\tilde{a}}=\sigma_{\K,a}  
\end{align}

We denote by \( \Pr_{S'}^{\K-1}[\tilde{a}_1,...,\tilde{a}_{\K-1} \mid \tilde{q}_1,...,\tilde{q}_{\K-1}] \) the probability that the first \( \K-1 \) classical provers which execute the strategy $S'$ return the answers \( \tilde{a} = (\tilde{a}_1, \dots, \tilde{a}_{\K-1}) \) upon receiving the questions \( (\tilde{q}_1,...,\tilde{q}_{\K-1}) \). Analogously, we denote by \( \Pr_{S}^{\K-1}[\tilde{a}_1,...,\tilde{a}_{\K-1} \mid \tilde{q}_1,...,\tilde{q}_{\K-1}] \) the corresponding probability for the first \( \K-1 \) out of the \( \K \) provers which execute the strategy \( S \). Theorem \ref{thm: multi_round_clifford_classical} proves the correctness of Algorithm~\ref{alg:classical_simulation_non_local_game} and hence guarantees that \( \Pr_{S'}^{\K-1}[\tilde{a}_1,...,\tilde{a}_{\K-1} \mid \tilde{q}_1,...,\tilde{q}_{\K-1}] =  \Pr_{S}^{\K-1}[\tilde{a}_1,...,\tilde{a}_{\K-1} \mid \tilde{q}_1,...,\tilde{q}_{\K-1}]                                         \).\newline 

Furthermore, \( \Pr_{S'}^{\K}[\tilde{a}_\K \mid \tilde{q}_1,...,\tilde{q}_{\K-1},\tilde{q}_\K, \tilde{a}_1,...,\tilde{a}_{\K-1}] \) denotes the probability that the last of the $\K$ provers executing the strategy $S'$ returns the answer \( \tilde{a}_\K \), conditioned on all provers receiving the questions \( \tilde{q}:=(\tilde{q}_1,...,\tilde{q}_\K) \) and the first \( \K-1 \) provers returning the answers \( \tilde{a} = (\tilde{a}_1, \dots, \tilde{a}_{\K-1}) \). \( \Pr_{S}^{\K}[\tilde{a}_\K \mid \tilde{q}_1,...,\tilde{q}_{\K-1},\tilde{q}_\K, \tilde{a}_1,...,\tilde{a}_{\K-1}] \) denotes the corresponding probability for the \( K \)-th prover executing \( S \). By the definition of the strategy of the \( K \)-th prover following Algorithm~\ref{alg:classical_simulation_non_local_game},

\begin{align}
\text{Pr}_{S'}^{\K}[\tilde{a}_\K \mid \tilde{q}_1,...,\tilde{q}_{\K-1},\tilde{q}_\K, \tilde{a}_1,...,\tilde{a}_{\K-1}] = \text{Pr}_{S}^{\K}[\tilde{a}_\K \mid q_1,...,q_{\K-1},\tilde{q}_\K, a_1,...,a_{\K-1}] 
\end{align}

Due to Eq.\ \eqref{eq:res_state_equality_1}, \( \Pr_{S}^{\K}[\tilde{a}_\K \mid q_1,...,q_{\K-1},\tilde{q}_\K, a_1,...,a_{\K-1}]  =\Pr_{S}^{\K}[\tilde{a}_\K \mid \tilde{q}_1,...,\tilde{q}_{\K-1},\tilde{q}_\K, \tilde{a}_1,...,\tilde{a}_{\K-1}] \). Therefore, 

\begin{align}
\text{Pr}_{S'}^{\K}[\tilde{a}_\K \mid \tilde{q}_1,...,\tilde{q}_{\K-1},\tilde{q}_\K, \tilde{a}_1,...,\tilde{a}_{\K-1}] = \text{Pr}_{S}^{K}[\tilde{a}_K \mid \tilde{q}_1,...,\tilde{q}_{\K-1},\tilde{q}_\K, \tilde{a}_1,...,\tilde{a}_{\K-1}] 
\end{align}

We can now state the equality which shows the correctness of Algorithm~\ref{alg:classical_simulation_non_local_game}:
\begin{align}
\Pr_{S}[\tilde{a}_1,...,\tilde{a}_\K \mid \tilde{q}_1,...,\tilde{q}_\K]=\Pr_{S}[\tilde{a}_1,...,\tilde{a}_{\K-1} \mid \tilde{q}_1,...,\tilde{q}_{\K}]\cdot \Pr_{S}[\tilde{a}_{\K} \mid \tilde{q}_1,...,\tilde{q}_{\K-1}, \tilde{q}_K, \tilde{a}_1,...,\tilde{a}_{\K-1}]\notag & = \\
\Pr_{S'}[\tilde{a}_1,...,\tilde{a}_{\K-1} \mid \tilde{q}_1,...,\tilde{q}_{\K}]\cdot \Pr_{S'}[\tilde{a}_{\K} \mid \tilde{q}_1,...,\tilde{q}_{\K-1}, \tilde{q}_\K, \tilde{a}_1,...,\tilde{a}_{\K-1}]=\Pr_{S'}[\tilde{a}_1,...,\tilde{a}_\K \mid \tilde{q}_1,...,\tilde{q}_\K]
\end{align}

Hence, $S,S'$ are equivalent strategies. Due to the fact that Algorithm~\ref{alg:classical_simulation_non_local_game} requires classically simulating a general quantum computation in step 3, we do not provide a non-trivial bound on its
running-time complexity.
\end{proof}

\end{document}